\documentclass[11pt]{article}
\usepackage{fullpage}
\usepackage{amssymb, amsthm, amsmath}
\usepackage{tikz}
\usepackage{xifthen} 		
\usepackage{subfigure}	
\usepackage{hyperref}
\usepackage[noblocks]{authblk}

\tikzstyle{black node}=[draw,circle,fill=black, minimum size=3pt, inner sep=0pt]	
\newtheorem{thm}{Theorem}[section]
\newtheorem{cor}[thm]{Corollary}

\theoremstyle{remark}
\newtheorem{rem}[thm]{Remark}
\theoremstyle{definition}

\author{Matthew P. Johnson}
\affil{Pennsylvania State University}

\author{Den{\.{i}}z Sar{\i}{\"{o}}z\thanks{Research supported by grants from NSA (47149-0001) and \mbox{PSC-CUNY} (63427-0041).}}
\affil{The Graduate School and University Center of The City University of New York}

\title{Computing the Obstacle Number of a Plane Graph}

\begin{document}
\maketitle
\thispagestyle{empty}
\begin{abstract}
An \emph{obstacle representation} of a plane graph $G$ is $V(G)$ together with a set of opaque polygonal obstacles such that $G$ is the visibility graph on $V(G)$ determined by the obstacles.  
We investigate the problem of computing an obstacle representation of a plane graph (ORPG) with a minimum number of obstacles.
We call this minimum size the \emph{obstacle number} of $G$.

First, we show that ORPG is NP-hard by reduction from planar vertex cover, resolving a question posed by \cite{SariozCCCG11}.   
Second, we give a reduction from ORPG to maximum degree 3 planar vertex cover. Since this reduction preserves solution values, it follows that ORPG is fixed parameter tractable (FPT) and admits a polynomial-time approximation scheme (PTAS).
\end{abstract}

\section{Introduction}
Let $G$ be a plane graph with straight edges and vertices in general position; that is, a straight-line drawing of a planar graph with no edge crossings and no three vertices on a line in which the vertices are identified with their positions.
We refer to the open line segment between a pair of non-adjacent graph vertices as a \emph{non-edge} of $D$.
An \emph{obstacle representation} of $G$ is a pair $(V(G), \mathcal{O})$ where $\mathcal{O}$ is a set of polygons (not necessarily convex) called \emph{obstacles}, 
such that:
\begin{enumerate}
\item
$G$ does not meet any obstacle, and
\item
every non-edge of $G$ meets at least one obstacle.
\end{enumerate}

Equivalently,
$G$ is the visibility graph on $V(G)$ determined by the obstacles in $\mathcal{O}$.  
The size of an obstacle representation is the cardinality of $\mathcal{O}$.
Denote by \emph{ORPG} the problem of computing a minimum-size obstacle representation of $G$ (the optimum of which is called the {\em obstacle number} of $G$).
Alpert, Koch, and Laison introduced the notions \emph{obstacle representation} and \emph{obstacle number} for abstract graphs \cite{AKL09}
and noted that in any minim\emph{al} obstacle representation, each obstacle can be identified with the face it lies in.
Hence, we will use the terms \emph{face} and \emph{obstacle} interchangeably.
If the faces have weights
then we can seek a minimum-weight obstacle representation.

Finding a minimum-size obstacle representation of a straight-line graph drawing 
was treated as a computational problem
in the setting in which $D$ and $G$ need not be planar  \cite{SariozCCCG11}.
This problem was reduced to hypergraph transversal (hitting set),
with $O(n^4)$ faces available to pierce $O(n^2)$ non-edges 
($O(n)$ faces and $\Theta(n^2)$ non-edges in the ORPG special case).
A randomized $O(\log OPT)$-approximation algorithm 
based on bounding the Vapnik-Chervonenkis dimension of the corresponding hypergraph family was given in \cite{SariozCCCG11}.
Left open was the question of whether better approximations or perhaps optimal algorithms were feasible.

In this note we give partial answers to that question. We show that computing the obstacle number is NP-hard already in the special case of plane graphs; nonetheless, we show that ORPG admits a polynomial-time approximation scheme (PTAS) and is fixed-parameter tractable (FPT). 
We show hardness by a reduction from planar vertex cover; the positive results are consequences of a solution value-preserving reduction to maximum degree 3 planar vertex cover.

\section{Reduction from planar vertex cover}

\begin{thm}
ORPG is NP-hard.
\label{thm:ORPG_NPcomplete}
\end{thm}


\begin{proof}
We reduce from planar vertex cover. 
Recall that in the decision version of planar vertex cover, 
we are given an abstract planar graph $G$ having (without loss of generality) no isolated vertex, and a number $k$.
Let $n = |V(G)|$, $m = |E(G)|$, and denote by $f$ the number of faces in any crossing-free planar drawing of $G$.
We will transform $G$ in polynomial time 
into a plane graph $G'$ in such a way that $G$ has a vertex cover of size $k$ if and only if $G'$ has an obstacle representation of size $k'$ (for $k'$ defined below).

First, we construct from the planar vertex cover instance $G$ a planar vertex cover problem instance $G^3$ with maximum degree 3, 
adapting and extending the construction of \cite{GJ77}. 
The graph $G^3$ admits a vertex cover of size $k'$ if and only if $G$ admits a vertex cover of size $k$. Second, we construct an ORPG instance $G'$ in such a way that an obstacle representation of $G'$ will correspond to a vertex cover of $G^3$ of the same size, and vice versa.

\vskip.2cm\noindent \textbf{Constructing the maximum degree 3 planar vertex cover instance $G^3$.}
The planar graph $G^3$ is constructed as follows. We transform each vertex $v_i$ of $G$ into a cycle $C^{i}$ of length $2b_i$, with $b_i \in \deg(v_i) + \{0,1,2\}$ (with the exact value decided below). 
We color the vertices of $C^{i}$ alternating between blue and red. We then create a single leaf vertex $z_i$ adjacent to some arbitrary red vertex of $C^{i}$. We transform each edge $(v_i, v_j)$ of $G$ into a path $P_{ij}$ with \emph{three} edges whose endpoints are \emph{distinct} blue vertices of $C^{i}$ and $C^{j}$.
We finally create $f$ copies of the 3-vertex path graph $P_3$, each constituting a component of $G^3$.


We claim that $G$ has a vertex cover of size at most $k$ if and only if $G^3$ has one of size at most $k' = k + f + m + \sum_i b_i$. 

($\Rightarrow$): 
For each vertex $v_i$ in a given vertex cover for $G$ of size $k$, we select $z_i$ and all the blue vertices of $C^{i}$, thus including an endpoint of each path $P_{ij}$; and for each $v_i$ not in the cover, we select all the red vertices of $C^{i}$ (a total so far of $k + \sum_i b_i$ vertices). 
Since for every path $P_{ij}$ at least one of the cycles $C^{i}$ and $C^{j}$ will have all its blue vertices chosen, thus including at least one endpoint of $P_{ij}$, choosing one internal vertex from each $P_{ij}$ ($m$ more), and the central vertex of each $P_3$ ($f$ more) suffices to complete a size $k'$ vertex cover for $G^3$.

($\Leftarrow$): 
Given a vertex cover for $G^3$ of size $k'$, we obtain a canonical vertex cover for $G^3$ of size $k '' \leq k'$ in the following way.
Each copy of $P_3$ contributes at least one vertex to a cover, 
so have it contribute exactly its central vertex,
for a total of $f$ vertices.
Each path $P_{ij}$ contributes at least one of its internal vertices to cover its central edge.
If both internal vertices of a path $P_{ij}$ are in the given cover,
take one internal vertex out and ensure that its blue neighbor is in,
which makes for $m$ internal vertices from these paths.
Note that every cycle $C^i$ contributes at least $b_i$ vertices, 
lest some edge of the cycle be uncovered.
This holds with equality only if $C^i$ contributes (including `its' $z_i$) exactly its red vertices.
Otherwise, ensure that $C^i$ contributes exactly $1 + b_i$ vertices: `its' $z_i$ and its blue vertices.
Denote by $k''$ the size of this resulting canonical vertex cover.
The cycles in $G^3$ contributing blue vertices therefore correspond to a vertex cover for $G$ of size 
$k'' - f - m - \sum_i b_i 
\leq
k' - f - m - \sum_i b_i 
= 
k$.

\vskip.2cm\noindent \textbf{Constructing the ORPG instance $\tilde G$.}
In the remainder of the proof, we show how to ``implement'' the graph $G^3$ as an equivalent ORPG problem instance. The basic building blocks of the construction are \emph{empty triangles} and \emph{diamonds}. 
An \emph{empty triangle} is a face of a plane graph that is surrounded by three edges and has no vertex inside.
A \emph{diamond} consists of two empty triangles sharing an edge and having their \emph{four} vertices in convex position.
Observe that a diamond contains a non-edge between two of its vertices. 
Hence at least one empty triangle of every diamond must be chosen in an obstacle representation.
The $f$ copies of
$P_3$ in $G^3$ will match the faces of $\tilde G$ besides empty triangles, all of which must be chosen.
The remaining vertices of $G^3$ will match the empty triangles of $\tilde G$, such that the edges among them match the diamonds of $\tilde G$.
Hence there will be a natural bijection between vertex covers of $G^3$ and obstacle representations of $\tilde G$.

To begin the construction, we use the linear-time algorithm of de Fraysseix, Pach, and Pollack \cite{FPP90} to
obtain a planar imbedding of $G$ on a \mbox{$O(n) \times O(n)$} portion of the integer lattice and then perturb the coordinates to obtain general position. 
(We do not distinguish between $G$ and this imbedding.) 
We first visualize $\tilde G$ as a bold drawing \cite{vanKreveld10} of $G$, whose vertices are represented by small disks and edges by solid rectangles: we draw each vertex $u_i$ of $G$ as a disk $D_i$ about $u_i$ (with boundary $\tilde C^i$), 
and every edge $u_i u_j$ as a solid rectangle $R_{ij}$.  See Fig. \ref{fig:bolddraw}.
Each $R_{ij}$ has two vertices $t_{ij}$,$v_{ij}$ on $\tilde C^i$ 
and two vertices $t_ {ji},v_ {ji}$ on $\tilde C^{j}$
such that the line $u_i u_j$ is a midline of $R_{ij}$,
and $t_{ij} u_i v_{ij} t_{ji} u_j v_{ji}$ is a counterclockwise ordering of the vertices of a convex hexagon.

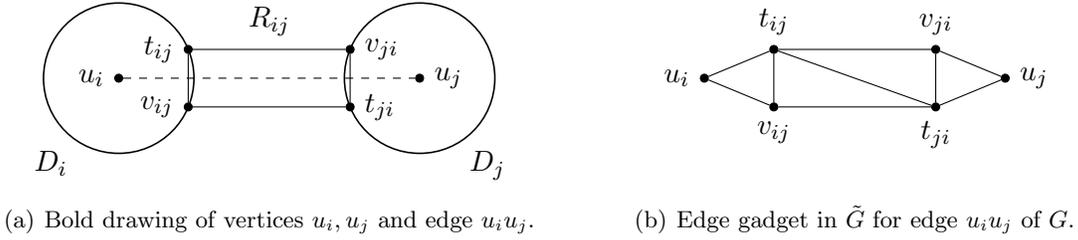
\begin{figure*}[t!]
\begin{center}
\subfigure[Bold drawing of vertices $u_i, u_j$ and edge $u_iu_j$.]{
\begin{tikzpicture}[scale=4]
\newcommand{\radius}{0.25}	        			
\newcommand{\halfAlpha}{22.5} 				
\newcommand{\thecos}{{cos(\halfAlpha)}}
\newcommand{\thesin}{{sin(\halfAlpha)}}
\newcommand{\rightx}{1}
\coordinate (P) at (-\radius * 1.5, -\radius * 1.5);
\coordinate (Q) at (1.5 * \radius + \rightx, \radius * 1.5);
\draw[white] (P) rectangle (Q);		
\node [black node] 	(ui) 	at (0,0) 		[label=left:$u_i$] {};
\node [black node] 	(uj) 	at (\rightx,0) 	[label=right:$u_j$] {};
\draw [semithick] 	(ui) 	circle (\radius);
\draw [semithick] 	(uj) 	circle (\radius);
\node	 [black node]	(tij)	at ({\halfAlpha}:{\radius}) 			[label=left:$t_{ij}$] {};
\node	 [black node]	(vij)	at ({-\halfAlpha}:{\radius}) 			[label=left:$v_{ij}$] {};
\begin{scope}[shift={(uj)}]
\node	 [black node]	(tji)	at ({\halfAlpha}:{-\radius}) 			[label=right:$t_{ji}$] {};
\node	 [black node]	(vji)	at ({-\halfAlpha}:{-\radius})			[label=right:$v_{ji}$] {};
\end{scope}
\draw [dashed] (ui) -- (uj);
\draw (tij) -- (vji) -- (tji) -- (vij) -- (tij);
\node 	(Ci)   at  ({-0.9*{\radius}}, {-0.7*{\radius}})  		[label=below:$D_i$] {}; 
\node 	(Cj)   at  ({\rightx + 0.9*{\radius}}, {-0.7*{\radius}})  	[label=below:$D_j$] {}; 
\node 	(Rij)   at  ({0.5*{\rightx}}, {\radius * {\thesin - 0.02}})  		[label=above:$R_{ij}$] {}; 
\end{tikzpicture}
\label{fig:bolddraw}
}
\quad
\subfigure[Edge gadget in $\tilde G$ for edge $u_iu_j$ of $G$.]{
\begin{tikzpicture}[scale=4]
\newcommand{\radius}{0.25}	        			
\newcommand{\halfAlpha}{22.5} 				
\newcommand{\thecos}{{cos(\halfAlpha)}}
\newcommand{\thesin}{{sin(\halfAlpha)}}
\newcommand{\rightx}{1}
\coordinate (P) at (-\radius * 1.5, -\radius * 1.5);
\coordinate (Q) at (1.5 * \radius + \rightx, \radius * 1.5);
\draw[white] (P) rectangle (Q);			
\node [black node] 	(ui) 	at (0,0) 		[label=left:$u_i$] {};
\node [black node] 	(uj) 	at (\rightx,0) 	[label=right:$u_j$] {};
\node	 [black node]	(tij)	at ({\halfAlpha}:{\radius}) 			[label=above :$t_{ij}$] {};
\node	 [black node]	(vij)	at ({-\halfAlpha}:{\radius}) 			[label=below :$v_{ij}$] {};
\begin{scope}[shift={(uj)}]
\node	 [black node]	(tji)	at ({\halfAlpha}:{-\radius}) 			[label=below :$t_{ji}$] {};
\node	 [black node]	(vji)	at ({-\halfAlpha}:{-\radius})			[label=above :$v_{ji}$] {};
\end{scope}
\node 	(Rij)   at  ({0.5*{\rightx}}, {\radius * {\thesin - 0.02}})  		{}; 
\draw (tij) -- (vij);
\draw (tji) -- (vji);
\draw (tij) -- (vji);
\draw (tji) -- (vij);
\draw (ui) -- (tij);
\draw (ui) -- (vij);
\draw (uj) -- (tji);
\draw (uj) -- (vji);
\draw (tij) -- (tji);
\end{tikzpicture}
\label{fig:GprimeForSingleEdge}
}
\end{center}
\caption{Bold drawing and edge gadget for an edge of $G$.}
\end{figure*}

We draw the disks small enough to ensure that they are well-separated from one another.
We set the radius $r$ of every disk to the smaller of 1/4 and half of the minimum distance between 
a vertex $u_i$ and an edge $u_j u_k$ ($j \neq i \neq k$) of $G$. 
To fix a single width for all rectangles (i.e., $||t_{ij} - v_{ij}||$), 
we set a global angle measure $\alpha$ to the smaller of $45^{\circ}$ and 
half of the smallest angle between two edges of $E(G)$ incident on the same vertex of $V(G)$.

$\tilde G$ is modeled on the bold drawing, by implementing each edge of $G$ (path $P_{ij}$ of $G^3$) with an edge gadget and each vertex of $G$ (cycle $C^{i}$ of $G^3$) with a vertex gadget. The edge gadget, consisting of four triangles forming three diamonds, is shown in Fig. \ref{fig:GprimeForSingleEdge}. (Note that each pair $v_{ij} v_{ji}$ defines a non-edge.) 

The vertex gadget is a modified wheel graph whose triangles correspond to the vertices of cycles $C^i$ in $G^3$ 
(see Fig. \ref{fig:nodegadget}).
On every circle $\tilde C^i$, for every edge $u_i u_j$ in $G$, we color blue the arc of measure $\alpha$
centered about the intersection of circle $\tilde C^i$ with $u_i u_j$ (a non-edge in $\tilde G$).
We place $t_{ij}$ and $v_{ij}$ at the endpoints of this arc 
so that
$t_{ij} u_{i} v_{ij}$ is a counterclockwise triple.
By the choice of $\alpha$, all blue arcs are well-separated, and hence the rectangles are well-separated from one another and from other disks, by the choice of $r$. We color the remaining arcs red to obtain a red-blue striped pattern on each circle $\tilde C^i$, corresponding in color to the vertices of the corresponding $C^i$ in $G^3$.

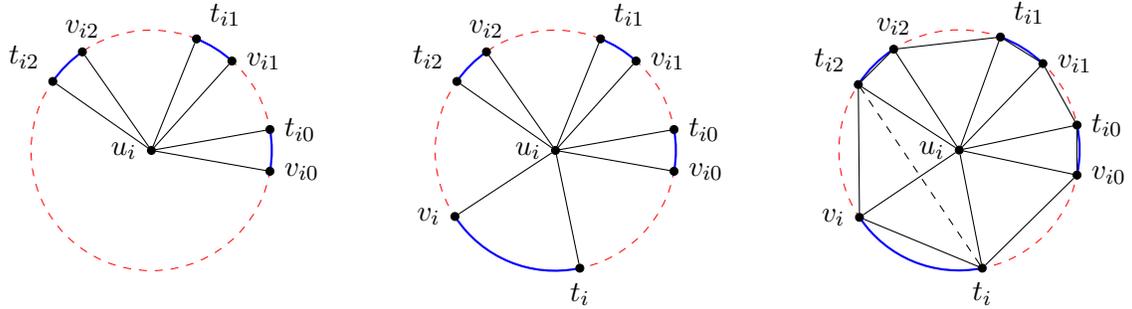
\begin{figure*}[t!]
\begin{center}
\subfigure[Initial circle with blue (solid) arcs of measure $\alpha$ and red (dashed) arcs has a large red arc of measure in $[180^{\circ} - \alpha$, $270^{\circ})$.]{
\begin{tikzpicture}[scale=0.80]
\newcommand{\radius}{2}	        			
\newcommand{\halfAlpha}{10} 				
\newcommand{\thecos}{{cos(\halfAlpha)}}
\newcommand{\thesin}{{sin(\halfAlpha)}}
\newcommand{\leftCenterToLeftEdge}{0.8}
\newcommand{\centersToTopEdge}{\leftCenterToLeftEdge * 0.9}
\newcommand{\rightx}{\radius * 2 + \radius * 3}
\coordinate (P) at (-\radius - \leftCenterToLeftEdge, -\radius - \centersToTopEdge);
\coordinate (Q) at (\radius + \leftCenterToLeftEdge, \radius + \centersToTopEdge);
\draw[white] (P) rectangle (Q);		
\node [black node] 	(ui) 	at (0,0) 		[label=left:$u_i$] {};
\draw [red, very thin, dashed] 	(ui) 	circle (\radius);
\coordinate	 (ti0)	at ({\halfAlpha}:{\radius}) ;
\coordinate (vi0)	at ({-\halfAlpha}:{\radius}) ;
\newcommand{\vtpairOneMidpointElevation}{58}
\newcommand{\tOneElevation}{(\vtpairOneMidpointElevation + \halfAlpha)}
\newcommand{\vOneElevation}{(\vtpairOneMidpointElevation - \halfAlpha)}
\coordinate	(ti1)	at ({\tOneElevation}:{\radius});
\coordinate (vi1)	at ({\vOneElevation}:{\radius});
\newcommand{\vtpairTwoMidpointElevation}{135}
\newcommand{\tTwoElevation}{(\vtpairTwoMidpointElevation + \halfAlpha)}
\newcommand{\vTwoElevation}{(\vtpairTwoMidpointElevation - \halfAlpha)}
\coordinate (ti2)	at ({\tTwoElevation}:{\radius}) ;	
\coordinate (vi2)	at ({\vTwoElevation}:{\radius}) ;	
\draw[thick, blue] (vi0) arc 	(-\halfAlpha:\halfAlpha:\radius);
\node	 [black node]	(ti0)	at (ti0) 					[label=right:		$t_{i0}$] 	{};
\node	 [black node]	(vi0)	at ({-\halfAlpha}:{\radius}) 		[label=right:		$v_{i0}$] 	{};
\draw[thick, blue] (vi1) arc 	(\vtpairOneMidpointElevation-\halfAlpha:\vtpairOneMidpointElevation+\halfAlpha:\radius);
\node	 [black node]	(ti1)	at (ti1) 					[label=above right:	$t_{i1}$] 	{};
\node	 [black node]	(vi1)	at (vi1) 					[label=right:		$v_{i1}$] 	{};
\draw[thick, blue] (vi2) arc 	(\vtpairTwoMidpointElevation-\halfAlpha:\vtpairTwoMidpointElevation+\halfAlpha:\radius);
\node	 [black node]	(ti2)	at (ti2) 					[label=above left:	$t_{i2}$] 	{};
\node	 [black node]	(vi2)	at (vi2) 					[label=above:		$v_{i2}$] 	{};
\draw (ui) -- (ti0);
\draw (ui) -- (vi0);
\draw (ui) -- (ti1);
\draw (ui) -- (vi1);
\draw (ui) -- (ti2);
\draw (ui) -- (vi2);

\newcommand{\measureOfRedArc}{(360 - \halfAlpha - \vtpairTwoMidpointElevation - \halfAlpha)}
\newcommand{\measureOfRedSubArc}{(\measureOfRedArc / 3)}
\end{tikzpicture}
}
\quad
\subfigure[After subdividing the large red arc into \emph{three}, coloring its middle part blue, and adding dummy $v_i$ and $t_i$ vertices.]{
\begin{tikzpicture}[scale=0.80]
\newcommand{\radius}{2}	        			
\newcommand{\halfAlpha}{10} 				
\newcommand{\thecos}{{cos(\halfAlpha)}}
\newcommand{\thesin}{{sin(\halfAlpha)}}
\newcommand{\leftCenterToLeftEdge}{0.8}
\newcommand{\centersToTopEdge}{\leftCenterToLeftEdge * 0.9}
\newcommand{\rightx}{\radius * 2 + \radius * 3}
\coordinate (P) at (-\radius - \leftCenterToLeftEdge, -\radius - \centersToTopEdge);
\coordinate (Q) at (\radius + \leftCenterToLeftEdge, \radius + \centersToTopEdge);
\draw[white] (P) rectangle (Q);		
\node [black node] 	(ui) 	at (0,0) 		[label=left:$u_i$] {};
\draw [red, very thin, dashed] 	(ui) 	circle (\radius);
\coordinate	 (ti0)	at ({\halfAlpha}:{\radius}) ;
\coordinate (vi0)	at ({-\halfAlpha}:{\radius}) ;
\newcommand{\vtpairOneMidpointElevation}{58}
\newcommand{\tOneElevation}{(\vtpairOneMidpointElevation + \halfAlpha)}
\newcommand{\vOneElevation}{(\vtpairOneMidpointElevation - \halfAlpha)}
\coordinate	(ti1)	at ({\tOneElevation}:{\radius});
\coordinate (vi1)	at ({\vOneElevation}:{\radius});

\newcommand{\vtpairTwoMidpointElevation}{135}
\newcommand{\tTwoElevation}{(\vtpairTwoMidpointElevation + \halfAlpha)}
\newcommand{\vTwoElevation}{(\vtpairTwoMidpointElevation - \halfAlpha)}
\coordinate (ti2)	at ({\tTwoElevation}:{\radius}) ;	
\coordinate (vi2)	at ({\vTwoElevation}:{\radius}) ;	

\draw (ui) -- (ti0);
\draw (ui) -- (vi0);
\draw (ui) -- (ti1);
\draw (ui) -- (vi1);
\draw (ui) -- (ti2);
\draw (ui) -- (vi2);

\draw[thick, blue] (vi0) arc 	(-\halfAlpha:\halfAlpha:\radius);
\node	 [black node]	(ti0)	at (ti0) 					[label=right:		$t_{i0}$] 	{};
\node	 [black node]	(vi0)	at ({-\halfAlpha}:{\radius}) 		[label=right:		$v_{i0}$] 	{};

\draw[thick, blue] (vi1) arc 	(\vtpairOneMidpointElevation-\halfAlpha:\vtpairOneMidpointElevation+\halfAlpha:\radius);
\node	 [black node]	(ti1)	at (ti1) 					[label=above right:	$t_{i1}$] 	{};
\node	 [black node]	(vi1)	at (vi1) 					[label=right:		$v_{i1}$] 	{};

\draw[thick, blue] (vi2) arc 	(\vtpairTwoMidpointElevation-\halfAlpha:\vtpairTwoMidpointElevation+\halfAlpha:\radius);
\node	 [black node]	(ti2)	at (ti2) 					[label=above left:	$t_{i2}$] 	{};
\node	 [black node]	(vi2)	at (vi2) 					[label=above:		$v_{i2}$] 	{};

\newcommand{\measureOfRedArc}{(360 - \halfAlpha - \vtpairTwoMidpointElevation - \halfAlpha)}
\newcommand{\measureOfRedSubArc}{(\measureOfRedArc / 3)}

\coordinate (vidummy0) 			at ({\tTwoElevation + \measureOfRedSubArc}:{\radius}) ;
\coordinate (tidummy0)			at ({\tTwoElevation + \measureOfRedSubArc + \measureOfRedSubArc}:{\radius}) ;
\draw[thick, blue] (vidummy0) arc 	(\tTwoElevation + \measureOfRedSubArc : \tTwoElevation + \measureOfRedSubArc + \measureOfRedSubArc : \radius);
\node	 [black node]	(vidummy0) 	at (vidummy0) 		[label=left:		$v_{i}$] {};
\node	 [black node]	(tidummy0)	at (tidummy0) 		[label=below:	$t_{i}$] {};

\draw[thin] (ui) -- (vidummy0);
\draw[thin] (ui) -- (tidummy0);

\end{tikzpicture}
\label{fig:subdivideVeryLargeRedArc}
}
\quad
\subfigure[In the resulting wheel graph, each pair of triangles sharing an edge induce a diamond.  One diamond's non-edge is shown dashed.]{
\begin{tikzpicture}[scale=0.80]
\newcommand{\radius}{2}	        			
\newcommand{\halfAlpha}{12} 				
\newcommand{\thecos}{{cos(\halfAlpha)}}
\newcommand{\thesin}{{sin(\halfAlpha)}}
\newcommand{\leftCenterToLeftEdge}{0.8}
\newcommand{\centersToTopEdge}{\leftCenterToLeftEdge * 0.9}
\newcommand{\rightx}{\radius * 2 + \radius * 3}
\coordinate (P) at (-\radius - \leftCenterToLeftEdge, -\radius - \centersToTopEdge);
\coordinate (Q) at (\radius + \leftCenterToLeftEdge, \radius + \centersToTopEdge);
\draw[white] (P) rectangle (Q);		
\node [black node] 	(ui) 	at (0,0) 		[label=left:$u_i$] {};
\draw [red, very thin, dashed] 	(ui) 	circle (\radius);
\coordinate	 (ti0)	at ({\halfAlpha}:{\radius}) ;
\coordinate (vi0)	at ({-\halfAlpha}:{\radius}) ;
\newcommand{\vtpairOneMidpointElevation}{58}
\newcommand{\tOneElevation}{(\vtpairOneMidpointElevation + \halfAlpha)}
\newcommand{\vOneElevation}{(\vtpairOneMidpointElevation - \halfAlpha)}
\coordinate	(ti1)	at ({\tOneElevation}:{\radius});
\coordinate (vi1)	at ({\vOneElevation}:{\radius});

\newcommand{\vtpairTwoMidpointElevation}{135}
\newcommand{\tTwoElevation}{(\vtpairTwoMidpointElevation + \halfAlpha)}
\newcommand{\vTwoElevation}{(\vtpairTwoMidpointElevation - \halfAlpha)}
\coordinate (ti2)	at ({\tTwoElevation}:{\radius}) ;	
\coordinate (vi2)	at ({\vTwoElevation}:{\radius}) ;	

\draw (ui) -- (ti0);
\draw (ui) -- (vi0);
\draw (ui) -- (ti1);
\draw (ui) -- (vi1);
\draw (ui) -- (ti2);
\draw (ui) -- (vi2);

\draw[thick, blue] (vi0) arc 	(-\halfAlpha:\halfAlpha:\radius);
\node	 [black node]	(ti0)	at (ti0) 					[label=right:		$t_{i0}$] 	{};
\node	 [black node]	(vi0)	at ({-\halfAlpha}:{\radius}) 		[label=right:		$v_{i0}$] 	{};

\draw[thick, blue] (vi1) arc 	(\vtpairOneMidpointElevation-\halfAlpha:\vtpairOneMidpointElevation+\halfAlpha:\radius);
\node	 [black node]	(ti1)	at (ti1) 					[label=above right:	$t_{i1}$] 	{};
\node	 [black node]	(vi1)	at (vi1) 					[label=right:		$v_{i1}$] 	{};

\draw[thick, blue] (vi2) arc 	(\vtpairTwoMidpointElevation-\halfAlpha:\vtpairTwoMidpointElevation+\halfAlpha:\radius);
\node	 [black node]	(ti2)	at (ti2) 					[label=above left:	$t_{i2}$] 	{};
\node	 [black node]	(vi2)	at (vi2) 					[label=above:		$v_{i2}$] 	{};

\newcommand{\measureOfRedArc}{(360 - \halfAlpha - \vtpairTwoMidpointElevation - \halfAlpha)}
\newcommand{\measureOfRedSubArc}{(\measureOfRedArc / 3)}

\coordinate (vidummy0) 			at ({\tTwoElevation + \measureOfRedSubArc}:{\radius}) ;
\coordinate (tidummy0)			at ({\tTwoElevation + \measureOfRedSubArc + \measureOfRedSubArc}:{\radius}) ;
\draw[thick, blue] (vidummy0) arc 	(\tTwoElevation + \measureOfRedSubArc : \tTwoElevation + \measureOfRedSubArc + \measureOfRedSubArc : \radius);
\node	 [black node]	(vidummy0) 	at (vidummy0) 		[label=left:		$v_{i}$] {};
\node	 [black node]	(tidummy0)	at (tidummy0) 		[label=below:	$t_{i}$] {};

\draw[thin] (ui) -- (vidummy0);
\draw[thin] (ui) -- (tidummy0);

\draw[thin] (vi0) -- (ti0) -- (vi1) -- (ti1) -- (vi2) -- (ti2) -- (vidummy0) -- (tidummy0) -- (vi0);

\draw[dashed] (tidummy0) -- (ti2);	 

\end{tikzpicture}
\label{fig:subdivideVeryLargeRedArc}
}
\end{center}\label{fig:nodegadget}
\caption{Constructing the wheel graph drawing in the case of a large red arc.}
\end{figure*}

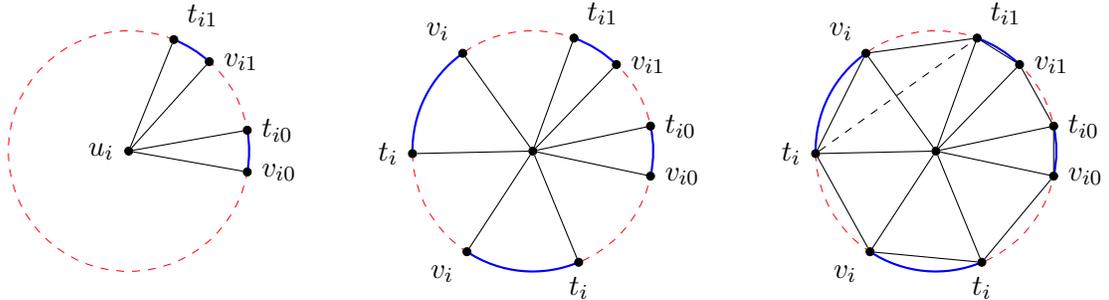
\begin{figure*}[t!]
\begin{center}
\subfigure[An initial circle with blue (solid) and red (dashed) arcs has a {very} large red arc, of measure at least $270^{\circ}$.]{
\begin{tikzpicture}[scale=0.80]
\newcommand{\radius}{2}	        			
\newcommand{\halfAlpha}{10} 				
\newcommand{\thecos}{{cos(\halfAlpha)}}
\newcommand{\thesin}{{sin(\halfAlpha)}}
\newcommand{\leftCenterToLeftEdge}{0.8}
\newcommand{\centersToTopEdge}{\leftCenterToLeftEdge * 0.9}
\newcommand{\rightx}{\radius * 2 + \radius * 3}
\coordinate (P) at (-\radius - \leftCenterToLeftEdge, -\radius - \centersToTopEdge);
\coordinate (Q) at (\radius + \leftCenterToLeftEdge, \radius + \centersToTopEdge);
\draw[white] (P) rectangle (Q);		
\node [black node] 	(ui) 	at (0,0) 		[label=left:$u_i$] {};
\draw [red, very thin, dashed] 	(ui) 	circle (\radius);
\coordinate	 (ti0)	at ({\halfAlpha}:{\radius}) ;
\coordinate (vi0)	at ({-\halfAlpha}:{\radius}) ;
\newcommand{\vtpairOneMidpointElevation}{58}
\newcommand{\tOneElevation}{(\vtpairOneMidpointElevation + \halfAlpha)}
\newcommand{\vOneElevation}{(\vtpairOneMidpointElevation - \halfAlpha)}
\coordinate	(ti1)	at ({\tOneElevation}:{\radius});
\coordinate (vi1)	at ({\vOneElevation}:{\radius});

\newcommand{\vtpairTwoMidpointElevation}{135}
\newcommand{\tTwoElevation}{(\vtpairTwoMidpointElevation + \halfAlpha)}
\newcommand{\vTwoElevation}{(\vtpairTwoMidpointElevation - \halfAlpha)}
\coordinate (ti2)	at ({\tTwoElevation}:{\radius}) ;	
\coordinate (vi2)	at ({\vTwoElevation}:{\radius}) ;	

\draw[thick, blue] (vi0) arc 	(-\halfAlpha:\halfAlpha:\radius);
\node	 [black node]	(ti0)	at (ti0) 					[label=right:		$t_{i0}$] 	{};
\node	 [black node]	(vi0)	at ({-\halfAlpha}:{\radius}) 		[label=right:		$v_{i0}$] 	{};

\draw[thick, blue] (vi1) arc 	(\vtpairOneMidpointElevation-\halfAlpha:\vtpairOneMidpointElevation+\halfAlpha:\radius);
\node	 [black node]	(ti1)	at (ti1) 					[label=above right:	$t_{i1}$] 	{};
\node	 [black node]	(vi1)	at (vi1) 					[label=right:		$v_{i1}$] 	{};

\draw (ui) -- (ti0);
\draw (ui) -- (vi0);
\draw (ui) -- (ti1);
\draw (ui) -- (vi1);

\newcommand{\measureOfRedArc}{(360 - \halfAlpha - \vtpairTwoMidpointElevation - \halfAlpha)}
\newcommand{\measureOfRedSubArc}{(\measureOfRedArc / 3)}


\end{tikzpicture}
}
\quad
\subfigure[After subdividing the very large red arc into \emph{five}, coloring its second and fourth parts blue, and adding dummy $v_i$ and $t_i$ vertices.]{
\begin{tikzpicture}[scale=0.80]
\newcommand{\radius}{2}	        			
\newcommand{\halfAlpha}{12} 				
\newcommand{\thecos}{{cos(\halfAlpha)}}
\newcommand{\thesin}{{sin(\halfAlpha)}}
\newcommand{\leftCenterToLeftEdge}{0.8}
\newcommand{\centersToTopEdge}{\leftCenterToLeftEdge * 0.9}
\newcommand{\rightx}{\radius * 2 + \radius * 3}
\coordinate (P) at (-\radius - \leftCenterToLeftEdge, -\radius - \centersToTopEdge);
\coordinate (Q) at (\radius + \leftCenterToLeftEdge, \radius + \centersToTopEdge);
\draw[white] (P) rectangle (Q);		
\node [black node] 	(ui) 	at (0,0)	 {};	
\draw [red, very thin, dashed] 	(ui) 	circle (\radius);
\coordinate	 (ti0)	at ({\halfAlpha}:{\radius}) ;
\coordinate (vi0)	at ({-\halfAlpha}:{\radius}) ;
\newcommand{\vtpairOneMidpointElevation}{58}
\newcommand{\tOneElevation}{(\vtpairOneMidpointElevation + \halfAlpha)}
\newcommand{\vOneElevation}{(\vtpairOneMidpointElevation - \halfAlpha)}
\coordinate	(ti1)	at ({\tOneElevation}:{\radius});
\coordinate (vi1)	at ({\vOneElevation}:{\radius});

\draw (ui) -- (ti0);
\draw (ui) -- (vi0);
\draw (ui) -- (ti1);
\draw (ui) -- (vi1);

\draw[thick, blue] (vi0) arc 	(-\halfAlpha:\halfAlpha:\radius);
\node	 [black node]	(ti0)	at (ti0) 					[label=right:		$t_{i0}$] 	{};
\node	 [black node]	(vi0)	at ({-\halfAlpha}:{\radius}) 		[label=right:		$v_{i0}$] 	{};

\draw[thick, blue] (vi1) arc 	(\vtpairOneMidpointElevation-\halfAlpha:\vtpairOneMidpointElevation+\halfAlpha:\radius);
\node	 [black node]	(ti1)	at (ti1) 					[label=above right:	$t_{i1}$] 	{};
\node	 [black node]	(vi1)	at (vi1) 					[label=right:		$v_{i1}$] 	{};


\newcommand{\measureOfRedArc} {(360 - \halfAlpha - \vtpairOneMidpointElevation - \halfAlpha)}
\newcommand{\measureOfRedSubArc} {(\measureOfRedArc / 5)}

\coordinate (vidummy0) 			at ({\tOneElevation + \measureOfRedSubArc}:{\radius}) ;
\coordinate (tidummy0)			at ({\tOneElevation + \measureOfRedSubArc + \measureOfRedSubArc}:{\radius}) ;
\draw[thick, blue] (vidummy0) arc 	({\tOneElevation + \measureOfRedSubArc} : {\tOneElevation + \measureOfRedSubArc + \measureOfRedSubArc} : \radius);
\node	 [black node]	(vidummy0) 	at (vidummy0) 		[label=above left:	$v_{i}$] {};
\node	 [black node]	(tidummy0)	at (tidummy0) 		[label=left:			$t_{i}$] {};

\coordinate (vidummy1) 			at ({\tOneElevation + 3*\measureOfRedSubArc}:{\radius}) ;
\coordinate (tidummy1)			at ({\tOneElevation + 3*\measureOfRedSubArc + \measureOfRedSubArc}:{\radius}) ;
\draw[thick, blue] 	(vidummy1) arc 	(\tOneElevation + 3*\measureOfRedSubArc : \tOneElevation + 3*\measureOfRedSubArc + \measureOfRedSubArc : \radius);
\node	 [black node]	(vidummy1) 	at (vidummy1) 		[label=below left:	$v_{i}$] {};
\node	 [black node]	(tidummy1)	at (tidummy1) 		[label=below:	$t_{i}$] {};

\draw[thin] (ui) -- (vidummy0);
\draw[thin] (ui) -- (tidummy0);
\draw[thin] (ui) -- (vidummy1);
\draw[thin] (ui) -- (tidummy1);

\end{tikzpicture}
}
\quad
\subfigure[In the resulting wheel graph, each pair of triangles sharing an edge induce a diamond.  One diamond's non-edge is shown dashed.]{
\begin{tikzpicture}[scale=0.80]
\newcommand{\radius}{2}	        			
\newcommand{\halfAlpha}{12} 				
\newcommand{\thecos}{{cos(\halfAlpha)}}
\newcommand{\thesin}{{sin(\halfAlpha)}}
\newcommand{\leftCenterToLeftEdge}{0.8}
\newcommand{\centersToTopEdge}{\leftCenterToLeftEdge * 0.9}
\newcommand{\rightx}{\radius * 2 + \radius * 3}
\coordinate (P) at (-\radius - \leftCenterToLeftEdge, -\radius - \centersToTopEdge);
\coordinate (Q) at (\radius + \leftCenterToLeftEdge, \radius + \centersToTopEdge);
\draw[white] (P) rectangle (Q);		
\node [black node] 	(ui) 	at (0,0)	 {};	
\draw [red, very thin, dashed] 	(ui) 	circle (\radius);
\coordinate	 (ti0)	at ({\halfAlpha}:{\radius}) ;
\coordinate (vi0)	at ({-\halfAlpha}:{\radius}) ;
\newcommand{\vtpairOneMidpointElevation}{58}
\newcommand{\tOneElevation}{(\vtpairOneMidpointElevation + \halfAlpha)}
\newcommand{\vOneElevation}{(\vtpairOneMidpointElevation - \halfAlpha)}
\coordinate	(ti1)	at ({\tOneElevation}:{\radius});
\coordinate (vi1)	at ({\vOneElevation}:{\radius});

\draw (ui) -- (ti0);
\draw (ui) -- (vi0);
\draw (ui) -- (ti1);
\draw (ui) -- (vi1);

\draw[thick, blue] (vi0) arc 	(-\halfAlpha:\halfAlpha:\radius);
\node	 [black node]	(ti0)	at (ti0) 					[label=right:		$t_{i0}$] 	{};
\node	 [black node]	(vi0)	at ({-\halfAlpha}:{\radius}) 		[label=right:		$v_{i0}$] 	{};

\draw[thick, blue] (vi1) arc 	(\vtpairOneMidpointElevation-\halfAlpha:\vtpairOneMidpointElevation+\halfAlpha:\radius);
\node	 [black node]	(ti1)	at (ti1) 					[label=above right:	$t_{i1}$] 	{};
\node	 [black node]	(vi1)	at (vi1) 					[label=right:		$v_{i1}$] 	{};


\newcommand{\measureOfRedArc} {(360 - \halfAlpha - \vtpairOneMidpointElevation - \halfAlpha)}
\newcommand{\measureOfRedSubArc} {(\measureOfRedArc / 5)}

\coordinate (vidummy0) 			at ({\tOneElevation + \measureOfRedSubArc}:{\radius}) ;
\coordinate (tidummy0)			at ({\tOneElevation + \measureOfRedSubArc + \measureOfRedSubArc}:{\radius}) ;
\draw[thick, blue] (vidummy0) arc 	({\tOneElevation + \measureOfRedSubArc} : {\tOneElevation + \measureOfRedSubArc + \measureOfRedSubArc} : \radius);
\node	 [black node]	(vidummy0) 	at (vidummy0) 		[label=above left:	$v_{i}$] {};
\node	 [black node]	(tidummy0)	at (tidummy0) 		[label=left:			$t_{i}$] {};

\coordinate (vidummy1) 			at ({\tOneElevation + 3*\measureOfRedSubArc}:{\radius}) ;
\coordinate (tidummy1)			at ({\tOneElevation + 3*\measureOfRedSubArc + \measureOfRedSubArc}:{\radius}) ;
\draw[thick, blue] 	(vidummy1) arc 	(\tOneElevation + 3*\measureOfRedSubArc : \tOneElevation + 3*\measureOfRedSubArc + \measureOfRedSubArc : \radius);
\node	 [black node]	(vidummy1) 	at (vidummy1) 		[label=below left:	$v_{i}$] {};
\node	 [black node]	(tidummy1)	at (tidummy1) 		[label=below:	$t_{i}$] {};

\draw[thin] (ui) -- (vidummy0);
\draw[thin] (ui) -- (tidummy0);
\draw[thin] (ui) -- (vidummy1);
\draw[thin] (ui) -- (tidummy1);

\draw[thin] (vi0) -- (ti0) -- (vi1) -- (ti1) -- (vidummy0) -- (tidummy0) -- (vidummy1) -- (tidummy1) -- (vi0);

\draw[dashed] (tidummy0) -- (ti1);	 

\end{tikzpicture}
}
\end{center}
\caption{Constructing the wheel graph drawing in the case of a \emph{very} large red arc.}
\label{fig:subdivideLargeRedArc}
\end{figure*}

On every circle $\tilde C^i$,
we will
add the remaining edges between consecutive vertices of $\tilde C^i$
to complete the union of the triangles $t_{ij} u_i v_{ij}$, forming  a wheel graph on hub $u_{i}$,
such that every pair of triangles sharing a spoke form a diamond.
If a red arc has measure at least $180^{\circ} - \alpha$, however, we must add additional spokes.
By the general position assumption, at most one red arc per wheel can have such great measure.
If such a red arc has measure less than $270 ^{\circ}$, we divide it evenly into \emph{three} parts and color the middle part blue
(see Fig. \ref{fig:nodegadget});
otherwise, we divide it evenly into \emph{five} parts and color the second and fourth parts blue 
(see Fig. \ref{fig:subdivideLargeRedArc}), 
maintaining the striped pattern in both cases.
We place dummy $t_{i}$ and $v_{i}$ vertices\footnote{Dummy vertices have no adjacencies with any vertices outside of $D_i$.}
at the newly created 
(\emph{zero}, \emph{two} or \emph{four}) 
arc endpoints.
Finally, we add the requisite edges to complete the wheel graph.

We place a vertex $\tilde z_i$ on an arbitrary red arc of $\tilde C^i$ and connect it in $\tilde G$ to the end vertices (say $t_{ij}$ and $v_{ik}$) of that arc. 
Thus an empty triangle $t_{ij} z_i v_{ik}$ is formed in $\tilde G$ 
as part of a diamond with $u_i$, corresponding to $z_i$ and its incident edge in $G^3$.

In the unbounded face of $\tilde G$ we place two isolated vertices inducing a non-edge inside the unbounded face, thus requiring this face to be chosen in any solution.
Every non-triangular face of $\tilde G$ must be selected as an obstacle, 
since every simple polygon with at least 4 vertices has an internal diagonal (i.e., a non-edge).
The selection of these faces are forced moves and correspond to the selection, in a vertex cover for $G^3$, of the central vertex of each $P_3$.

This completes the construction of $\tilde G$. Since each pair of neighboring triangles in $\tilde G$ indeed form a diamond and every non-triangular face is indeed a forced move, the result follows.
\end{proof}

\begin{rem}
To represent coordinates \emph{exactly} as described would require a very permissive 
unit-cost RAM model of computation in which it is possible to represent real numbers and perform arithmetic and trigonometric functions in unit time. The reduction above can be modified in such a way that each vertex position of $\tilde G$ is represented using $O(\log n)$ bits.
\end{rem}

In an alternate proof strategy, we can begin with a special {\em touching polygons} representation of $G$ instead of a bold drawing. This strategy would involve topologically ``collapsing'' each rectangle $R_{ij}$ to a single edge shared by the wheel graphs on hubs $u_i$ and $u_j$, which is the unique edge crossing segment $u_i u_j$. 
There is a linear-time algorithm for computing a touching polygons representation where the number of sides in a polygon is at most \emph{six} \cite{Kobourov11}, this algorithm can be modified to ensure that every side of a polygon is an edge of $\tilde G$.
The linear-time algorithm for computing a touching polygons representation given in \cite{Kobourov11} can easily be modified to ensure that every edge of $G$ corresponds to a distinct side between two polygons,
but using this method there is no guaranteed way to place hubs inside their corresponding polygons such that every pair of adjacent polygons have hubs that define non-edges that meet the shared polygon side.
Nonetheless, a polynomial-time algorithm by Mohar \cite{Mohar97} does ensure this.

\section{Reduction \emph{to} vertex cover}

\begin{thm}
Weighted ORPG is reducible to weighted maximum degree 3 planar vertex cover by an optimal solution value-preserving reduction.
\label{reductionToP-VC-3}
\end{thm}
\begin{proof}
Given a plane graph $G$ on $n$ vertices in general position, we construct a graph $\hat{G}$ that admits a vertex cover of cost $k$ if and only if $G$ admits an obstacle representation of cost $k$.

Every bounded non-triangular face of $G$ must be selected as an obstacle; moreover, the unbounded face must be chosen if and only if its convex hull boundary contains a non-edge. Since these are forced moves, we henceforth assume without loss of generality that every non-edge we must block meets at least two faces. 

Recall that an \emph{empty triangle} is a bounded face on three vertices not containing any other vertices, and that a \emph{diamond} consists of two empty triangles that share an edge and have their four vertices in convex position.

We claim that every non-edge must meet the two triangles forming some diamond, and hence must meet those triangles' shared edge.
Assume for contradiction that some remaining non-edge $s$ never crosses the diagonal edge of a diamond.
Denote by $u$ and $v$ the endpoints of $s$, and orient the plane such that $u$ is directly below $v$.
Obtain a sequence of empty triangles $(f_0, f_1, \ldots, f_k)$ by tracing $s$ 
from $u$ (a vertex on $f_0$) 
to $v$ (a vertex on $f_k$).
Denote by $v_i$ (for $1 \leq i \leq k$) the unique vertex in face $f_k$ that is not a vertex of $f_{i-1}$ (so that $v_k = v$).
Without loss of generality, the reflex angle of $f_0$ and $f_1$ is to the right of $s$, which implies that $v_1$ is to the right of $s$.
In order for $f_2$ to be the next face in this sequence, $v_2$ must be to the left of $s$.
In general, in order for $f_i$ to be the next face in this sequence, 
$v_i$ must be on the other side of $s$ from $v_{i-1}$.
This pattern must continue indefinitely, lest two consecutive triangles form a diamond. The indefinite continuation of this pattern implies an infinite sequence of faces defined by $s$, and hence a contradiction.

We now define $\hat{G}$, which is a subgraph of the dual of $G$: each edge of $\hat G$ corresponds to diamond of $G$. 
The graph $\hat G$ is induced by these edges (with vertex weights set to the correspond face weights). 
For each diamond, at least one its two triangles must be chosen in any obstacle representation. 
Thus every obstacle representation of $\hat G$ corresponds to a vertex cover of $G$ of the same cost, and vice versa.
\end{proof}

\begin{rem}
We may wish to adopt the more realistic bit model, since a plane graph drawing may have been expressed using a number of bits super-polynomial in $n$ for vertex coordinates.
In this model, the reduction would require time super-polynomial in $|V(G)|$ but nonetheless polynomial in the number of input bits used for representing $G$.
\label{rem:WhatPolynomial}
\end{rem}

From Theorem \ref{reductionToP-VC-3} we immediately obtain the following.

\begin{cor}
Weighted ORPG admits a polynomial-time approximation scheme (PTAS) \cite{Baker94PTAS}.
\end{cor}

Moreover, this also follows.
\begin{cor}
ORPG is fixed parameter tractable (FPT).
\end{cor}
\begin{proof}
Perform the reduction of Theorem \ref{reductionToP-VC-3}, producing a planar maximum degree 3 vertex cover instance $\hat{G}$ with $V(\hat{G}) = \hat{n}$

Using the FPT algorithm by Xiao \cite{Xiao10_FPT} for maximum degree 3 vertex cover on $\hat{G}$,
we can compute an obstacle representation for $G$ with $k$ obstacles in additional time at most
$1.1616^{k - |F_0|} \hat{n}^{O(1)}
= 1.1616^{k - |F_0|} n^{O(1)}
= 1.1616^{k} n^{O(1)}$.

Alternatively, using the FPT algorithm by Alber et al. \cite{AFN01} for planar vertex cover on $\hat{G}$, we can
to compute an obstacle representation for $G$ with $k$ obstacles in additional time at most
$O(2^{4 \sqrt{3 (k - |F_0|)}} \hat{n})
= O(2^{4 \sqrt{3 (k - |F_0|)}} n)
= O(2^{4 \sqrt{3 k}} n)$.
\end{proof}

\bibliographystyle{plain}
\end{document}